\documentclass[11pt]{article}
\usepackage[margin=1in]{geometry}
\usepackage{times}
\usepackage{wrapfig}
\usepackage{citesort}
\usepackage{epsfig}
\usepackage{graphicx}
\usepackage{color}
\usepackage{fullpage}
\usepackage{amssymb}
\usepackage{amsmath}
\usepackage{theorem}
\usepackage[ruled, vlined, linesnumbered]{algorithm2e}
\usepackage{color}

\newcommand{\eps}{\epsilon}
\newcommand{\E}{\mathbf{E}}

\newcommand{\abs}[1]{\left| #1 \right|}
\newcommand{\norm}[1]{\left\lVert #1 \right\rVert}
\newcommand{\var}[1]{\mathbf{Var}\left[#1\right]}
\newcommand{\vol}{\text{vol}}

\newcommand{\sk}[1]{\text{sk}(#1)}

\newcommand{\windeg}[1]{\delta_{#1}^{\tt{in}}(\vec{G})}
\newcommand{\woutdeg}[1]{\delta_{#1}^{\tt{out}}(\vec{G})}
\newcommand{\uwindeg}[1]{d_{#1}^{\tt{in}}(\vec{G})}
\newcommand{\uwoutdeg}[1]{d_{#1}^{\tt{out}}(\vec{G})}

\newcommand{\windegP}[1]{\delta_{#1}^{\tt{in}}(\vec{P})}
\newcommand{\woutdegP}[1]{\delta_{#1}^{\tt{out}}(\vec{P})}

\newcommand{\uwoutdegP}[1]{d_{#1}^{\tt{out}}(\vec{P})}
\renewcommand{\L}{\mathcal{L}}
\renewcommand{\S}{\mathcal{S}}

\newcommand{\IND}{Indexing}

\newtheorem{theorem}{Theorem}

\newtheorem{lemma}{Lemma}

\newtheorem{observation}{Observation}

{\theorembodyfont{\rmfamily}
}
{\theorembodyfont{\rmfamily}
\newtheorem{definition}{Definition}}

\newcommand{\qinhides}[1]{}

\newenvironment{proof}{\trivlist\item[]\emph{Proof:}}%
{\unskip\nobreak\hskip 1em plus 1fil\nobreak$\Box$
\parfillskip=0pt%
\endtrivlist}

\begin{document}

\title{A Sketching Algorithm for Spectral Graph Sparsification}

\author{Jiecao Chen\thanks{Indiana University Bloomington, jiecchen@indiana.edu}, Bo Qin \thanks{CSE Department, Hong Kong University of Science and Technology, bqin@cse.ust.hk}, David P. Woodruff \thanks{IBM Almaden Research, dpwoodru@us.ibm.com}, Qin Zhang \thanks{Indiana University Bloomington,  qzhangcs@indiana.edu}}
  
\begin{titlepage}
\date{}

\maketitle

 \begin{abstract}
 We study the problem of compressing a weighted graph $G$ on $n$ vertices, building a ``sketch'' $H$ of $G$, so that given any vector $x \in \mathbb{R}^n$, the value $x^T L_G x$ can be approximated up to a multiplicative $1+\eps$ factor from only $H$ and $x$, where $L_G$ denotes the Laplacian of $G$. One solution to this problem is to build a spectral sparsifier $H$ of $G$, which, using the result of Batson, Spielman, and Srivastava, consists of $O(n \eps^{-2})$ reweighted edges of $G$ and has the property that simultaneously for all $x \in \mathbb{R}^n$, $x^T L_H x = (1 \pm \eps) x^T L_G x$. The $O(n \eps^{-2})$ bound is optimal for spectral sparsifiers. We show that if one is interested in only preserving the value of $x^T L_G x$ for a {\it fixed} $x \in \mathbb{R}^n$ (specified at query time) with high probability, then there is a sketch $H$ using only $\tilde{O}(n \eps^{-1.6})$ bits of space. This is the first data structure achieving a sub-quadratic dependence on $\eps$. Our work builds upon recent work of Andoni, Krauthgamer, and Woodruff who showed that $\tilde{O}(n \eps^{-1})$ bits of space is possible for preserving a fixed {\it cut query} (i.e., $x\in \{0,1\}^n$) with high probability; here we show that even for a general query vector $x \in \mathbb{R}^n$, a sub-quadratic dependence on $\eps$ is possible. Our result for Laplacians is in sharp contrast to sketches for general $n \times n$ positive semidefinite matrices $A$ with $O(\log n)$ bit entries, for which even to preserve the value of $x^T A x$ for a fixed $x \in \mathbb{R}^n$ (specified at query time) up to a $1+\eps$ factor with constant probability, we show an $\Omega(n \eps^{-2})$ lower bound. 

 \end{abstract}

 \end{titlepage}

\section{Introduction}
Given an $n \times n$ matrix $A$, an important task is to be able to compress
$A$ to a sketch $s(A)$ so that given $s(A)$, one can approximate
$\|Ax\|_2^2$ for a given query vector $x$. Typically, the form of approximation
is allowing the sketch to report a $(1+\eps)$-approximation to the value
of $\|Ax\|_2^2$, namely, a number $v(x) \in [(1-\eps)\|Ax\|_2^2, (1+\eps)\|Ax\|_2^2]$. One also can make the distinction of the sketch succeeding on all
queries $x$ simultaneously, or for any fixed query $x$ with high probability. In case of the former
we say $A$ has the ``for all'' guarantee, while in the latter case we say it has the ``for each''
guarantee. 
Ideally, $s(A)$ can be stored with
much less space than storing the original matrix $A$. 

Compressing a matrix so as to preserve its behavior on query vectors $x$
is a fundamental task and has a number of applications. For example, in 
graph sparsification 
$A$ is the square-root of the Laplacian $L(G)$ of a graph $G$ 
and $x$ is a cut or spectral query. In the case that $x \in \{0,1\}^n$
specifies a cut, $\|Ax\|_2^2$ gives the number of edges (or sum
of edge weights if $G$ is weighted) crossing the cut. 
The problem also arises in numerical linear algebra,
where $A$ is the adjoined matrix $[B, c]$ of a least squares
regression problem and querying the vector $(x, -1)$ 
gives the cost of the residual solution $\|Bx-c\|_2^2$. The fact
that $s(A)$ can be stored with small space leads to communication
savings in a distributed setting and memory savings in the data
stream model, see \cite{akw14} for communication savings 
for minimum cut and \cite{cw09} for memory savings for regression 
problems. 

For general $n \times n$ matrices $A$ with $O(\log n)$-bit entries, 
it is impossible for $s(A)$ to have the for all guarantee unless it uses
$\Omega(n^2)$ bits of space.  
One way to see this is to consider a net of $n \times n$ projector matrices
$P$ onto $n/2$-dimensional subspaces of $\mathbb{R}^n$. 
It is known that there
exists a family $\mathcal{F}$ of $r = 2^{\Omega(n^2)}$ distinct matrices 
$P_{1}, \ldots, P_{r}$ so that for all $i \neq j$, 
$\|P_i - P_j\|_2 \geq \frac{1}{2}$, where for a symmetric matrix $A$, 
$\|A\|_2 = \sup_x \frac{\|Ax\|}{\|x\|}$. An analogous bound holds even if we
round the entries of $P \in \mathcal{F}$ to finite precision. We can show then that 
if we had a sketch $s(A)$ for
each $A \in \mathcal{F}$ 
which had the for all guarantee, then we could uniquely recover $A \in \mathcal{F}$; 
see Theorem~\ref{thm:all-exact} in the Appendix for a formal proof. It follows that any sketch with the for all guarantee
requires $\Omega(\log |\mathcal{F}|) = \Omega(n^2)$ bits of space. 

In light of the above impossibility, two natural relaxations are possible: 
\begin{enumerate}
\item Consider
general $A$ but only allow the for each guarantee, 
\item Consider a restricted family
of matrices $A$ and allow the for each or for all guarantee, and possibly for a
restricted family of $x$. 
\end{enumerate}

Regarding (1), due to the Johnson-Lindenstrauss theorem, 
one can achieve a sketch $s(A)$ using only $O(n \eps^{-2} (\log 1/\delta) \log n)$
bits of space. One instantiation of this is 
that by choosing a random $r \times n$ matrix S with i.i.d. 
entries in $\{-1/\sqrt{r}, +1/\sqrt{r}\}$ for $r = \Theta(\eps^{-2} \log 1/\delta)$, for any fixed $x$,
$\|SAx\|_2^2 \in [(1-\eps)\|Ax\|_2^2, (1+\eps)\|Ax\|_2^2]$ with probability $1-\delta$. Therefore
the sketch $s(A)$ can just be $SA$. For general matrices $A$, this upper bound 
turns out to be optimal up to logarithmic factors. 
That is, one cannot
improve upon the $\Omega(n \eps^{-2})$ dependence. See Theorem~\ref{thm:all-approx} in the Appendix. Note that $\|Ax\|_2^2 = x^T A^T A x$ and $A^T A$ is positive semidefinite (PSD), so one cannot preserve $x^T B x$ with the for each guarantee even for PSD matrices $B$ with fewer than $\Omega(n \eps^{-2})$ bits.

Regarding (2), our understanding is much weaker. In this paper we will focus on the important 
case that $A$ is the square-root of a Laplacian $L(G)$ of a graph $G$ with polynomially bounded weights, and so $\|Ax\|_2^2 = x^T L(G) x$. 
Recently, Andoni, Krauthgamer, and Woodruff \cite{akw14}
show that even in the case that $x \in \{0,1\}^n$ is restricted to be a cut query, if $s(A)$
satisfies the for all guarantee then it must have $\Omega(n \eps^{-2})$ bits of space. This is matched,
(up to a logarithmic factor for the distinction of words versus bits) by the cut sparsifiers
of Batson, Spielman, and Srivastava \cite{BSS14}; see \cite{BK96,SS11,BSS14,FHHP11,KP12} 
for a sample of prior work on cut sparsifiers. 

Perhaps surprisingly, if instead one allows $s(A)$ to have the for each guarantee, then it is 
possible to achieve $n \eps^{-1} \cdot \textrm{polylog}(n)$ bits of space, so at least in the regime that 
$\eps \leq 1/\log^{O(1)} n$, this provides a substantial improvement over the for all case. This
bound holds for example even when $\epsilon = 1/\sqrt{n}$, in which case it gives
$n^{3/2} \cdot \textrm{polylog} n$ size instead of $\Omega(n^2)$ which is necessary
to preserve {\it all} cut values approximately. It is also shown in \cite{akw14} that $\Omega(n \eps^{-1})$ bits of
space is necessary for the for each guarantee. 

While the above gives a substantial improvement for cut queries, it does not quite fit the above framework
of sketching a restricted family of matrices to approximately preserve the norm of 
any query vector $x \in \mathbb{R}^n$, as it holds only for cut queries. The same result of Batson, Spielman,
and Srivastava \cite{BSS14} shows that in the for all model, $O(n \eps^{-2})$ words is achievable for arbitrary
$x \in \mathbb{R}^n$ via their construction of spectral sparsifiers. They show a lower bound for a certain
type of data structure, namely, a reweighted subgraph, of $\Omega(n \eps^{-2})$ edges, while an $\Omega(n \eps^{-2})$
bit lower bound follows for arbitrary data structures by the result in \cite{akw14}. 

A natural question, posed by \cite{akw14}, 
is if one can beat the $O(n \eps^{-2})$ barrier in the for each model for Laplacians
allowing arbitrary query vectors $x \in \mathbb{R}^n$. Besides being of theoretical interest, arbitrary
queries on Laplacians give more flexibility than cut queries alone. For example, if the graph corresponds
to a physical system, e.g., the edges correspond to springs, the Laplacian allows for 
approximating the total power of the
system for a given set of potentials on the vertices. Also, since 
$x^TL(G)x = \sum_{e \in G} w_e (x_i - x_j)^2$, where $w_e$ is the weight of edge $e$ in graph $G$, 
the Laplacian evaluated on a permutation 
$\pi$ of $\{1, 2, \ldots, n\}$ gives the average squared distortion of a line embedding of $G$. 
\\\\
{\bf Our Contributions.} In this paper we give the first $o(n \eps^{-2})$ space data structure in the for
each model for Laplacians for arbitrary queries. Namely, we show that there is a data structure $s(A)$
using $n \eps^{-1.6} \textrm{polylog}(n)$ bits of space that can be built in $\textrm{poly}(n)$
time, such that for any
fixed query vector $x \in \mathbb{R}^n$, the output of $s(A)$ on query vector $x$ is in the range
$[(1-\eps)x^TL(G) x, (1+\eps)x^T L(G) x]$ with probability $1-1/n$.
\\\\
{\it An $n \eps^{-1.66} \textrm{polylog}(n)$ upper bound.}
We first show how to acheive a space bound of $n \eps^{-1.66} \textrm{polylog}(n)$ 
bits, and then show that this can be improved 
further to $n \eps^{-1.6} \textrm{polylog}(n)$. 

We can first make several simplifying assumptions. The first is that the total number of edges is
$O(n \eps^{-2})$. For this, we can first compute a spectral sparsifier \cite{SS11,BSS14}. It is useful
to note that if all edges weights were between $1$ and $\textrm{poly}(n)$, then after spectral sparsification
the edge weights are between $1$ and $\textrm{poly}(n)$, for a possibly larger polynomial. Next, we 
can assume all edge weights are within a factor of
$2$. Indeed, by linearity of the Laplacian, if all edge weights are in $[1, \textrm{poly}(n)]$, then 
we can group the weights 
into powers of $2$ and sketch each subset of edges separately, incurring an $O(\log n)$ factor blowup
in space. Third, and most importantly, we assume the Cheeger constant $h_G$ of each resulting
graph $G = (V,E)$ satisfies $h_G \geq \eps^{1/3}$, where recall for a graph $G$,
$$h_G = \inf_{S \subset V} \Phi_G(S) = \inf_{S \subset V} \frac{w_G(S, \bar{S})}{\min(\vol_G(S), \vol_G(\bar{S}))},$$
where $w_G(S, \bar{S}) = \sum_{u \in S, v \in \bar{S}} w_{u,v}$ for edge weights $u,v$, and
$\vol_G(S) = \sum_{u \in S} \sum_{v \mid (u,v) \in E} w_{u,v}$. We can assume $h_G \geq \eps^{1/3}$
because if it were not, then by definition of $h_G$ there is a sparse cut, that is, 
$\Phi(S) \leq \eps^{1/3}$. We can find a sparse cut (a polylogarithmic approximation suffices), store all
sparse cut edges in our data structure, and remove them from the graph $G$. We can then recurse on
the two sides of the cut. By a charging argument similar to that in \cite{akw14}, we can bound the total
number of edges stored across all sparse cuts. 

As for the actual data structure achieving our $n \eps^{-1.66} \textrm{polylog}(n)$ upper bound, 
we first store the weighted degree 
$\delta_u(G) = \sum_{v: (u,v) \in E} w_{u,v}$ of each node; this step is also done in the data structure
in \cite{akw14}. A difference is that we now partition vertices into
``heavy'' and ``light'' classes $V_L$ and $V_H$, where $V_L$ contains those vertices whose weighted degree
exceeds a threshold, and light consists of the remaining vertices. We include all edges incident to light
edges in the data structure. The remaining edges have both endpoints heavy and for each heavy vertex, we
randomly sample about $\eps^{-5/3}$ of its neighboring heavy edges; edge ${u,v}$ is sampled with probability
$\frac{w_{u,v}}{\delta_u(G_H)}$ where $\delta_u(G_H)$ is the sum of weighted edges from the heavy vertex $u$
to neighboring heavy vertices $v$. 

For the estimation procedure, when expanding the Laplacian 
$x^T L x = \sum_{u \in V} \sum_{v \in V} (x_u - x_v)^2 w_{u,v}$, we obtain
$$x^T L x = \sum_{u \in V} \delta_u(G) x_u^2 - \sum_{u \in V_L, v \in V} x_u x_v w_{u,v}
- \sum_{u \in V_H, v \in V_L} x_u x_v w_{u,v} - \sum_{u \in V_H} \sum_{v \in V_H} x_u x_v w_{u,v},$$
and our data structure has the first three summations on the right exactly; the only estimation comes from
estimating $\sum_{u \in V_H} \sum_{v \in V_H} x_u x_v w_{u,v}$, which we use our sampled heavy edges for. 
Since we only have heavy edges in this summation, 
this reduces our variance. Given this, we are able to upper bound the variance
as roughly $\eps^{10/3} \|D^{1/2} x\|_2^4$, where $D$ is a diagonal matrix with the degrees of $G$ on the diagonal.
We can then upper bound this norm by relating it to the first non-zero eigenvalue $\lambda_1(\tilde{L})$ 
of the normalized Laplacian, which cannot be too small, since by Cheeger's inequality, 
$\lambda_1(\tilde{L}) \geq \frac{h_G^2}{2}$, and we have ensured that $h_G$ is large.
\\\\
{\it An $n \eps^{-1.6} \textrm{polylog}(n)$ bit upper bound.} The rough idea to improve the previous
algorithm is to partition the edges
of $G$ into more refined groups based on the degrees of their endpoints. More precisely, we classify edges
$e$ by the minimum degree of their two endpoints, call this number $m(e)$, 
and two edges $e, e'$ are in the same class if when we round 
$m(e)$ and $m(e')$ to the nearest power of $2$, we obtain the same value. We note that the total number of vertices
with degree in $\omega(\eps^{-2})$ is $o(n)$, since we are starting with a graph with only $O(n \eps^{-2})$ edges;
therefore, all edges $e$ with $m(e) = \omega(\eps^{-2})$ can be handled by applying our entire procedure recursively
on say, at most $n/2$ nodes. Thus, it suffices to consider $m(e) \leq \epsilon^{-2}$. 

The intuition now is that as $m(e)$ increases, the variance of our estimator decreases since the two endpoints
have even larger degree now and so they are even ``heavier'' than before. Hence, we need fewer edge samples
when processing a subgraph restricted to edges with large $m(e)$. On the other hand, a graph on edges $e$ 
for which every value of $m(e)$ is small simply cannot have too many edges; indeed, every edge is incident
to a low degree vertex. Therefore, when we partition the graph to ensure the Cheeger constant $h_G$ is small,
since there are fewer total edges (before we just assumed this number was upper bounded by $n \eps^{-2}$), now
we pay less to store all edges across sparse cuts. Thus, we can balance these two extremes, and doing so we
arrive at our overall $n \eps^{-1.6} \textrm{polylog}(n)$ bit space bound. 

Several technical challenges arise when performing this more refined partitioning. One is that when doing the
sparse cut partitioning to ensure the Cheeger constant is small, we destroy the minimum degree of endpoints of
edges in the
graph. Fortunately we can show that for our setting of parameters, the total number of edges removed along
sparse cuts is small, and so only a small number of vertices have their degree drop by more than a factor of $2$.
For these vertices, we can afford to store all edges incident to them directly, so they do not contribute
to the variance. Another issue that arises is that to have small variance, we would like to ``assign'' each
edge $\{u,v\}$ to one of the two endpoints $u$ or $v$. If we were to assign it to both, we would have higher
variance. This involves creating a companion
or ``buddy graph'' which is a directed graph associated with the original graph. This directed graph assists
us with the edge partitioning, and tells us which edges to potentially sample from which vertices. 


\section{Preliminaries}
\paragraph{Notations and definitions.}
Let $G = (V, E, w)$ be an undirected positively weighted graph, with weight function $w : V \times V \to \mathbb{R}^+ \cup \{0\}$. For any $e = (u, v) \in E$, denote $w_{u,v}(G)$ or $w_e(G)$ to be its weight. For any $e = (u, v) \not\in E$, set $w_{u,v} = w_e = 0$. Let $w_{\max}$ and $w_{\min}$ be the maximum and minimum weights of edges in $E$ respectively.

For a vertex $u \in V$, let $\delta_u(G)$ be the weighted degree of $u$ in $G$, i.e. $\delta_u(G) = \sum_{v \in V} w_{u,v}$, and let $d_u(G)$ be the unweighted degree of $u$ in $G$, i.e. $d_u(G) = \abs{\{v \in V\ |\ w_{u,v} > 0\}}$.

Let $\vec{G} = (V, \vec{E}, w)$ be a directed positively weighted graph. For a vertex $u \in V$, define weighted in- and out-degree $\windeg{u} = \sum_{v: (v, u) \in \vec{E}} w_{v, u}$, $\woutdeg{u} = \sum_{v: (u, v) \in \vec{E}} w_{u, v}$, and unweighted in- and out-degree $\uwindeg{u} = |\{v\ |\ (v, u) \in \vec{E}, w_{v, u} > 0\}|$, and $\uwoutdeg{u} = |\{v\ |\ (u, v) \in \vec{E}, w_{u, v} > 0\}|$,

For a vertex set $S \subset V$, let $G(S)$ be the vertex-induced subgraph of $G$, and $E(S)$ be the edge set of $G(S)$. And for an edge set $F \subset E$, let $G(F)$ be the edge-induced subgraph of $G$, and $V(F)$ be the vertex set of $G(F)$; definitions will be the same if edges are directed.

Let $L(G)$ be the unnormalized Laplacian of $G$, and let $\hat{L}(G)$ be the normalized Laplacian of $G$.  Let $\lambda_1(L)$ be the second smallest eigenvalue of the matrix $L$ (note that the smallest eigenvalue of the Laplacian $L$ is $\lambda_0(L)=0$).


We say a random variable $X$ is a $(1+\eps, \delta)$-approximation of $Y$ if $(1-\eps)Y \le X \le (1+\eps)Y$ with probability at least $1 - \delta$.

For simplicity, we use $\tilde{O}(f)$ to denote $f \cdot \text{poly}\log(\abs{V}\abs{E}/(\eps\delta))$. We assume $\abs{V} > 1/\eps$ since otherwise we can just store the whole graph.

We define spectral sketch of a graph as follows.

\begin{definition}[$(1+\eps, \delta)$-spectral-sketch]
We say a sketch of $G$ (denoted by $\sk{G}$) is a {\em $(1+\eps, \delta)$-spectral-sketch} of $G$ if there is a reconstruction function that given $\sk{G}$ and $x\in \mathbb{R}^{|V|}$,  outputs a $(1+\eps)$-approximation to $x^TL(G)x$ with probability at least $1 - \delta$.
\end{definition}

Note that spectral sketch is different from spectral sparsifier~\cite{BSS14} in that (1) spectral sketch is ``for each" (i.e., preserve the value of the Laplacian quadratic form $x^TL(G)x$ for each $x \in \mathbb{R}^{\abs{V}}$ with high probability) while spectral sparsifier is ``for all" (i.e., works for all $x \in \mathbb{R}^{\abs{V}}$ with high probability). Clearly, a spectral sparsifier is also a spectral sketch. And (2) spectral sparsifier is a subgraph of the original graph $G$, while spectral sketch is not necessarily a subgraph. 






We will need Cheeger's constant and Cheeger's inequality.

\begin{definition}[Cheeger's constant]
For any $S \subset V$, let $w_G(S,\overline{S}) = \sum_{u\in S, v\in \overline{S}} w_{u,v}$ be the weighted cut in $G$, and let $\vol_G(S) = \sum_{u\in S}\delta_u(G)$ be the weighted volume of $S$ in $G$. Let $\Phi_G(S) = \frac{w_G(S, \overline{S})}{\min\{\vol_G(S), \vol_G(\bar{S})\}}$ be the {\em conductance} of the cut $(S, \bar{S})$. We define {\em Cheeger's constant} to be $h_G = \inf_{S\subset V}\Phi_G(S)$.
\end{definition}

\begin{lemma}[Cheeger's inequality \cite{Cheeger70}]
\label{lem:cheeger}
Let $G = (V, E, w)$ be an undirected positively weighted graph. Let $\hat L$ be the normalized Laplacian of $G$. Let $h_G$ be the Cheeger's constant of graph $G$. The following inequality holds,
\begin{equation}
  \label{eq:cheeger}
  \lambda_1(\hat L) \geq \frac{h_G^2}{2}.
\end{equation}
\end{lemma}

\begin{lemma}
\label{lem:spielman}
Given an undirected positively weighted graph $G = (V, E, w)$ with ${w_{\max}}/{w_{\min}} = \text{poly}(|V|)$, there is an algorithm that takes $G$ as the input, and output a graph $\tilde G = (V, \tilde E, \tilde w)$ such that
  \begin{enumerate}
    \item $\tilde G$ is a $(1+\eps, \delta)$-spectral-sketch of $G$ of size $\tilde{O}({|V|}/{\eps^2})$ bits.
    \item ${\tilde w_{\max}}/{ \tilde w_{\min}}$ is bounded by $\text{poly} (|V|)$.
  \end{enumerate}
\end{lemma}
 
\begin{proof}
We simply run the spectral sparsification algorithm by Batson, Spielman and Srivastava \cite{BSS14}, which produces $\tilde{G}$ with size of $\tilde{O}({|V|}/{\eps^2})$ bits. For the latter property, assume it is not true,  we would be able to pick a graph cut in $\tilde{G}$ that could not be bounded by $\text{poly}(w_{\min})$ hence leads to a contradiction.



\end{proof}

\section{A Basic Sketching Algorithm}
In this section, we described a $(1+\eps, \delta)$-spectral-sketch of size $\tilde{O}(\abs{V} / \eps^{\frac{5}{3}})$.  Let $\alpha = c_\alpha \eps^{-\frac{5}{3}}$ be a parameter we will use in this section, where $c_\alpha > 0$ is a large enough constant.

We will start with an algorithm for a class of special graphs, and then extend it to general graphs.

\subsection{Special Graphs}
\label{sec:basic-simple}
In this section we consider a class of special graphs, defined as follows.

\begin{definition}[S1-graph]
We say an undirected weighted graph $G = (V, E, w)$ is an \emph{S1-graph} (reads ``simple type-$1$ graph'') if it satisfies the followings.
\begin{enumerate}
\item All weights $\{w_e\ |\ e \in E\}$ are within a factor of $2$,  i.e. for any $e \in E$, $w_e \in [\gamma, 2 \gamma)$ for some $\gamma > 0$.

\item $h_G > \alpha\eps^2 = c_\alpha \eps^{\frac{1}{3}}$.

\end{enumerate}  
\end{definition}

Let $\mathcal{S}(G) = \{v\in V\ |\ \delta_v \leq \gamma\alpha\}$, $\mathcal{L}(G) = \{v \in V\ |\ \delta_v > \gamma\alpha\}$. For $u\in \mathcal{L}(G)$, let $\delta_u^{\L}(G) = \sum_{v\in\mathcal{L}(G)}w_{u,v}$. We will omit ``$(G)$" when there is no confusion. The algorithm for sketching S1-graph is described in Algorithm~\ref{alg:sketch-basic-simple}.  When we say ``add an edge to the sketch" we always mean ``add the edge together with its weight".

\begin{algorithm}[t]
\KwIn{An S1-Graph $G = (V, E, w)$; a quality control parameter $\eps$}
\KwOut{a $(1+\eps, 0.01)$-spectral-sketch $\sk{G}$ of $G$ }
$\sk{G} \gets \emptyset$\;
Add $\{\delta_u\ |\ u \in V\}$ to $\sk{G}$\;
\For {$u \in \mathcal{S}$}{
  Add all of $u$'s adjacent edges to $\sk{G}$\;
}

\For {$u \in \mathcal{L}$}{
  Add $\delta_u^\L$ to $\sk{G}$\;
  $E_u \gets \{ (u, v)\ |\ v \in \L \}$\;
  Sample (with replacement) $\alpha$ edges from $E_u$, where each time the probability of sampling $e = (u, v)\in E_u$ is $p_e = {w_e}/{\delta_u^\L}$\;
Add the sampled $\alpha$ edges to $\sk{G}$\;
}

\Return{$\sk{G}$}\;
\caption{{\bf Spectral-S1}($G, \eps$)}
\label{alg:sketch-basic-simple}
\end{algorithm}

Let $Y_u^v$ be the random variable denoting the number of times edge $(u, v)$ is sampled at Line $8$ in Algorithm \ref{alg:sketch-basic-simple}. It is easy to see that 
\begin{equation}
\label{eq:b-1}
\E[Y_u^v] = \frac{\alpha w_{u,v}}{\delta_u^\L} \quad \text{and} \quad \var{Y_u^v} = \alpha \left(1 - \frac{w_{u,v}}{\delta_u^\L}\right)\frac{w_{u,v}}{\delta_u^\L} \leq \alpha \frac{w_{u,v}}{\delta_u^\L}.
\end{equation}
Given a vector $x \in \mathbb{R}^{|V|}$, we use the following expression as an estimator of $x^TLx$:
\begin{equation}
  \label{eq:estimator}
  I_G = \sum_{u \in V} \delta_u x_u^2 - \sum_{u \in \mathcal{S}}\sum_{v \in V} x_u x_v w_{u,v}
        - \sum_{u \in \mathcal{L}}\sum_{v \in \mathcal{S}}x_u x_v w_{u,v}
        - \sum_{u \in \mathcal{L}}\frac{\delta_u^\L}{\alpha} \sum_{v \in \mathcal{L}} x_u x_v Y_u^v.
\end{equation}

\begin{lemma}
  \label{lem:S1-Estimator}
 Let $G = (V, E, w)$ be an S1-Graph and $L = L(G)$ be the (unnormalized) Laplacian of $G$, then $I_G$ (defined in Equation (\ref{eq:estimator})) is an unbiased estimator of $x^TLx$. Furthermore, it gives a $(1 + \eps, 0.01)$-approximation to $x^TLx$.
\end{lemma}

\begin{proof}
Since $\E[Y_u^v] = \frac{\alpha w_{u,v}}{\delta_u^\L}$ (by (\ref{eq:b-1})), it is straightforward to show that 
$$\E[I_G] = \sum_{u \in V} \delta_u x_u^2 - \sum_{u \in \mathcal{S}}\sum_{v \in V} x_u x_v w_{u,v}
            - \sum_{u \in \mathcal{L}}\sum_{v \in V}x_u x_v w_{u,v}
            = \sum_{(u, v) \in E}  (x_u - x_v)^2 w_{u,v} = x^T L x. $$
Now let us compute the variance of $I_G$. Note that if $\var{I_G} = O\left(\eps^2 (x^TLx)^2\right)$, then by taking constant $c_\alpha$ in $\alpha = c_\alpha \eps^{-\frac{5}{3}}$ large enough, a Chebyshev's inequality immediately yields the lemma. The variance of $I_G$
\begin{eqnarray}
  \var{I_G}  & = & \var{\sum_{u\in\L}\frac{\delta_u^\L}{\alpha} \sum_{v\in\L} x_ux_vY_u^v} \nonumber \\
             & = &  \sum_{u\in\L}\frac{(\delta_u^\L)^2}{\alpha^2} \sum_{v\in\L} x_u^2x_v^2\var{Y_u^v} \nonumber \\
             &\leq & \sum_{u\in\L}\frac{(\delta_u^\L)^2}{\alpha^2} x_u^2\sum_{v\in\L} x_v^2\frac{\alpha w_{u,v}}{\delta_u^\L} \quad \quad \quad (\text{by (\ref{eq:b-1})}) \nonumber \\
             & = & \frac{1}{\alpha}\sum_{u\in\L}\delta_u^\L x_u^2\sum_{v\in\L} x_v^2 w_{u,v} \nonumber \\
             &\leq  & \frac{1}{\alpha}\sum_{u\in\L}\delta_u^\L x_u^2\sum_{v\in\L} x_v^2 \frac{2\delta_v}{\alpha} \quad \quad  (w_{u,v} \le 2 \gamma \le \frac{2\delta_v}{\alpha} \text{ by def. of S1-graph and def. of $\L$}) \nonumber \\
             &\leq & \frac{2}{\alpha^2} \sum_{u \in V} \delta_ux_u^2\sum_{v \in V} \delta_v x_v^2 \quad \quad \quad \quad (\delta_u^\L \le \delta_u \text{ by definitions}) \nonumber \\
             & = & \frac{2}{\alpha^2}\norm{D^{1/2}x}_2^4, \label{eq:c-1}
\end{eqnarray}
where $D = \text{diag}(\delta_1, \delta_2, \ldots, \delta_n)$.  The normalized Laplacian of $G$ can be written as $\hat L = D^{-\frac{1}{2}} L D^{-\frac{1}{2}}$. Define $\hat{x} = D^{1/2}x$, we have 
$$
  \|\hat{x}\|_2^2 = \hat{x}^T \hat{x} \leq \frac{1}{\lambda_1(\tilde L)}\hat{x}^T \tilde L \hat{x} 
                  \overset{by~(\ref{eq:cheeger})}{\leq} \frac{2}{h_G^2}(x^TLx) \overset{\text{property of S1-graph}}{<} \frac{2}{\alpha^2\eps^4}(x^TLx),
$$
which together with (\ref{eq:c-1}) gives $\var{I_G} < \frac{8}{\alpha^6\eps^8}(x^TLx)^2 = O\left (\eps^2(x^TLx)^2 \right )$.
\end{proof}

Using the standard ``median-trick'' (i.e. run $\log(1/\delta)$ independent estimators of $I_G$ in (\ref{eq:estimator}) and return the median of them), we can boost the success probability to $1 - \delta$ for any $\delta > 0$. We summarize our result for the S1-Graph $G$ in the following theorem.

\begin{theorem}
\label{thm:basic-simple}
There is a sketching algorithm which given an S1-graph, outputs a $(1+\eps, \delta)$-spectral-sketch of size $\tilde{O}(\abs{V}/\eps^{\frac{5}{3}})$.
\end{theorem}

\subsection{General Graphs}
Now let us extend our result to general positively weighted simple graphs $G = (V, E, w)$. We now require $w_{\max}/w_{\min} = \text{poly}(\abs{V})$.

%
The following observation is due to the linearity of Laplacian.
\begin{observation}
  \label{ob:Laplacian}
  Given any simple graph $G = (V, E, w)$,  let $L$ be its Laplacian. Let $E_1, E_2, \ldots, E_k$ be a disjoint partition of $E$, and let $G_i = (V, E_i, w)$. Let $L_i$ be the Laplacian of $G_i$. We have $x^TLx = \sum_{i=1}^kx^TL_ix$ for any $x\in \mathbb{R}^{|V|}$.
\end{observation}

Our high-level idea is to reduce general graphs to S1-graphs. Based on Observation~\ref{ob:Laplacian}, we can first partition the edge set $E$ into $E_1, \ldots, E_k\ (k = \Theta(\log \abs{V}))$ such that for any $e \in E_i$ we have $w_{e} \in [2^{i-1}w_{\min}, 2^iw_{\min})$, and then sketch each subgraph $G_i$ separately. Finally, at the time of a query, we simply add all estimators $I_{G_i}$ together. Thus it suffices to focus on a  graph with all weight $w_e \in [\gamma, 2\gamma) $ for some $\gamma > 0$.

We next partition each subgraph $G_i$ further so that each component $P$ satisfies $h_P \ge c_{\alpha}\eps^{\frac{1}{3}}$. Once this property is established, we can use Algorithm~\ref{alg:sketch-basic-simple} to sketch each $P$ separately. We describe this preprocessing step in Algorithm~\ref{alg:preprocessing}.

\begin{algorithm}[t]
\KwIn{A graph $G = (V, E, w)$ such that for any $e \in E$, $w_e \in [\gamma, 2\gamma)$; a parameter $h > 0$}
\KwOut{
A set $\mathcal{P}$ of edge disjoint components of $G$ such that for each $P \in \mathcal{P}$, $h_P > h$; and a graph $Q$ induced by the rest of the edges in $G$.
}
$\mathcal{P} \gets \{ G \}$, $Q \gets \emptyset$\;
\While {$\exists$ $P \in \mathcal{P}$ such that Cheeger's constant $h_P \leq h$} {
  Find an arbitrary cut $(S, \overline S)$ in $P$, such that $\Phi(S) \leq h$\;
  Replace $P$ with its two subgraphs $P(S)$ and $P(\overline{S})$ in $\mathcal{P}$\;
  Add all edges in the cut $(S, \overline S)$ into $Q$\;
}

\Return{$(\mathcal{P}$, $Q$)}\;
\caption{{\bf Preprocessing}($G, h$)}
\label{alg:preprocessing}
\end{algorithm}

The $Q$ returned by Algorithm~\ref{alg:preprocessing} is a set of cut edges we will literally keep. The following lemma bounds the size of $Q$. The proof is folklore, and we include it for completeness.
\begin{lemma}
  \label{lem:boundQ}
  For any positively weighted graph $G = (V, E, w)$ such that for any $e \in E$, $w_e \in [\gamma, 2\gamma)$ for some $\gamma > 0$, the number of edges of $Q$ returned by Algorithm \ref{alg:preprocessing} \textbf{Preprocessing}($G, h$) is bounded by $O(h |E| \log |E|)$.
\end{lemma}

\begin{proof}
In Algorithm \ref{alg:preprocessing}, we recursively split the graph $G = (V, E, w)$ into connected components until for every component $P = (V_P, E_P)$, its Cheeger constant $h_P$ = $\inf_{S\subset V_P}\Phi_C(S) > h$. Consider a single splitting step: We find a cut $(S, \bar{S})$ with $\vol_P(S) \le \vol_P(\bar{S})$ in $P$ such that $\Phi_P(S) = \frac{w_P(S, \bar{S})}{\vol_P(S)} \leq h$. We can think in this step each edge in $\vol_P(S)$ contributes at most $h$ edges to $Q$ on average, while edges in $\vol_P(\bar{S})$ contribute nothing to $Q$. We call $S$ the {\em Smaller-Subset}. 

By the definition of volume, we have $\vol_P(V_P) = \vol_P(S\cup\overline{S}) = \vol_P(S) + \vol_P(\overline{S}) \geq 2 \vol_P(S)$, and $\vol_G(V) = 2\abs{E}$. Thus in the whole recursion process, each edge will appear at most $O(\log \abs{E})$ times in Smaller-Subsets, hence will contribute at most $O(h \log \abs{E})$ edges to $Q$. Therefore the number of edges of $Q$ is bounded by $O(h \abs{E} \log \abs{E})$ words.
\end{proof}
  
Now we show the main algorithm for general graphs and analyze its performance. The algorithm is described in Algorithm~\ref{alg:sketch-basic-general}.

\begin{algorithm}[t]
\KwIn{$G = (V, E, w)$ with all weights in $[w_{\min}, w_{\max}]$; a quality control parameter $\eps$}
\KwOut{A $(1+\eps, 0.01)$-spectral-sketch $\sk{G}$ of $G$}
Edge-disjointly partition $G$ into $\mathcal{H} = \{H_1, \ldots, H_{k}\}$ s.t. all edges in $H_i$ have weights in $[2^{i-1}w_{\min}, 2^iw_{\min})$\;
$\sk{G} \gets \emptyset$\;
\ForEach{$H \in \mathcal{H}$} {
  $(\mathcal{P}, Q) \gets $ {\bf Preprocessing}($H, \alpha\eps^2$)\;
  Add $Q$ into $\sk{G}$\;
  \For {$P \in \mathcal{P}$} {
    Add \textbf{Spectral-S1}($P, \eps$) into $\sk{G}$\;
  }
}

\Return{$\sk{G}$}\;
\caption{{\bf Spectral-Basic}($G, \eps$)}
\label{alg:sketch-basic-general}
\end{algorithm}

The following lemma summarize the functionality of Algorithm~\ref{alg:sketch-basic-general}.

\begin{lemma}
\label{lem:basic-general}
Given a graph $G = (V, E, w)$, let $\sk{G} \gets$~\textbf{Spectral-Basic}($G, \eps$), then for any given $x \in \mathbb{R}^{|V|}$, $\sk{G}$ can be used to construct an unbiased estimator $I_G$ which gives a $(1+\eps, 0.01)$-approximation to $x^T L(G) x$. The sketch $\sk{G}$ uses $\tilde{O}(\eps^{\frac{1}{3}}|E| + |V|/\eps^{\frac{5}{3}})$ bits.
\end{lemma}

\begin{proof}
In Algorithm~\ref{alg:sketch-basic-general}, $G$ is partitioned into a set of edge disjoint components $\mathcal{P} = \{P_1, \ldots, P_t\}$, and we build a sketch $\sk{P_i}$ for each $P_i \in \mathcal{P}$ from which we can construct an unbiased estimator $I_{P_i}$ for $x^T L(P_i) x$ with variance bounded by $O(\eps^2 (x^T L(P_i) x)^2)$ according to Lemma~\ref{lem:S1-Estimator}. Moreover, we have stored all edges between these components; let $Q$ be the induced subgraph of these edges.  Our  estimator to $x^T L(G) x$ is
\begin{equation}
\label{eq:est-basic-general}
I_G = \sum_{i = 1}^t I_{P_i} + x^T L(Q) x,
\end{equation}
where $I_{P_i}$ defined in Equation~(\ref{eq:estimator}) is an unbiased estimator of $x^T L(P_i) x$.
By the linearity of Laplacian (Observation~\ref{ob:Laplacian}), $I_G$ is an unbiased estimator of $x^T L(G) x$. Now consider its variance:
\begin{eqnarray*}
\var{I_G} &=& \var{\sum_{1\leq i\leq t} I_{P_i} + x^T L(Q) x} \quad \text{(due to the independence of $P_i$'s)}\\
&\leq& O(\epsilon^2)\sum_{1\leq i \leq t}\left(x^T L(P_i) x\right)^2 \\
&\leq& O(\epsilon^2) \left(\sum_{1\leq i \leq t}x^T L(P_i) x+x^T L(Q) x \right)^2 \quad \text{($L(P_i)$ and $L(Q)$ are  positive semidefinite)}\\
&=& O(\epsilon^2) \left(x^T L(G) x\right)^2.
\end{eqnarray*}
The correctness follows from a Chebyshev's inequality.

Now we bound the size of the sketch $\sk{G}$. Consider a particular $H$ at Line $3$ of Algorithm~\ref{alg:sketch-basic-general}. For each $P = (V_P, E_P) \in \mathcal{P}$, the size of $\sk{P}$ by running {\em Spectral-S1$(P, \eps)$} at Line $7$ is bounded by $\tilde{O}(\abs{V_P}/\eps^{\frac{5}{3}})$ bits (Theorem~\ref{thm:basic-simple}); and for the remaining subgraph $Q = (V_Q, E_Q)$, $\sk{Q}$ is bounded by $\tilde{O}(\eps^{\frac{1}{3}} \abs{E_Q})$ bits (Lemma~\ref{lem:boundQ}). Thus 
\begin{eqnarray*}
\text{size}(\sk{P_i}) & = & \text{size}(\sk{Q}) + \sum_{P \in \mathcal{P}} \text{size}(\sk{P}) \\
& \le & \tilde{O}\left(\eps^{\frac{1}{3}} \abs{E_Q} + \sum_{P \in \mathcal{P}}  \abs{V_P} /\eps^{\frac{5}{3}}\right) \\
& \le & \tilde{O}\left(\eps^{\frac{1}{3}} \abs{E} + \abs{V}/\eps^{\frac{5}{3}}\right) \text{bits}. \quad \quad (\text{$\{P \in \mathcal{P}\}$ are vertex-disjoint})
\end{eqnarray*}
Since there are $k = \Theta(\log\abs{V})$ of $H_i$'s in $\mathcal{H}$, the size of $\sk{G}$ is bounded by
 $\tilde{O}\left(\abs{V}/\eps^{\frac{5}{3}} + \eps^{\frac{1}{3}} \abs{E}\right) \cdot \log\abs{V} = \tilde{O}\left(\abs{V}/\eps^{\frac{5}{3}} + \eps^{\frac{1}{3}} \abs{E}\right)$ bits.
\end{proof}

We conclude this section with the following theorem.
\begin{theorem}
\label{thm:basic-general}
There is a sketching algorithm which given an undirected positively weighted graph $G = (V, E, w)$ with $w_{\max}/w_{\min} = \text{poly}(\abs{V})$, outputs a $(1+\eps, \delta)$-spectral-sketch of size $\tilde{O}(\abs{V}/\eps^{\frac{5}{3}})$.
\end{theorem}

\begin{proof}
The algorithm is as follows: we first run the spectral sparsification algorithm in \cite{BSS14}, obtaining a graph $\tilde{G} = (V, \tilde{E}, \tilde{w})$. By Lemma~\ref{lem:spielman} we have $|\tilde{E}| = \tilde{O}(\abs{V}/\eps^2)$ and $\tilde{w}_{\max}/\tilde{w}_{\min} = \text{poly}(\abs{V})$. We then run Algorithm~\ref{alg:sketch-basic-general}, getting a $(1+\eps, 0.01)$-spectral-sketch of size 
$
\tilde{O}\left(\abs{V}/\eps^{\frac{5}{3}} + \eps^{\frac{1}{3}} |\tilde{E}|\right) = \tilde{O}(\abs{V} / \eps^{\frac{5}{3}})
$. We can again run $\log(1/\delta)$ independent estimators and return the median of them to boost the success probability to $1 - \delta$.
\end{proof}

\section{An Improved Sketching Algorithm}
\label{sec:improve}
In this section, we further reduce the space complexity of the sketch to $\tilde O( {|V|}/{\eps^{\frac{8}{5}}} )$. At a high level, such an improvement is achieved by partitioning the graph into more subgroups (compared with a hierarchical partition on weights, and $\mathcal{S}(G)$ and $\mathcal{L}(G)$ for each weight class in the basic approach), in each of which vertices have similar unweighted degrees {\em and} weighted degrees. An estimator based on a set of sampled edges from such groups will have smaller variance. This finer partition, however, will introduce a number of technical subtleties, as we will describe below.

We set the constant $\beta = c_\beta \eps^{-\frac{8}{5}}$ throughout this section where $c_\beta$ is a constant.

\subsection{Special Graphs}
We first consider a class of simple graphs. 
\begin{definition}[S2-graph]
\label{def:S2-graph}
We say an undirected weighted graph $G = (V, E, w)$ is an {\em S2-graph} (reads ``simple type-2 graph") if we can assign directions to its edges in a certain way, getting a directed graph $\vec{G} = (V, \vec{E}, w)$
satisfying the following.
\begin{enumerate}
\item All weights $\{w_e\ |\ e \in \vec{E}\}$ are within a factor of $2$, i.e., for any $e \in \vec{E}$, $w_e \in [\gamma, 2\gamma)$ for some constant $\gamma > 0$.

\item For each $u \in V$, $\uwoutdeg{u} \in [2^\kappa\beta, 2^{\kappa+1}\beta)$, where $2^\kappa\beta \le 1/\eps^2$.

\end{enumerate}
We call $\vec{G}$ the {\em buddy} of $G$. Note that $\vec{G}$ is not necessarily unique, and we just need to consider an arbitrary but fixed one. In this section we will assume that we can obtain such a buddy directed graph $\vec{G}$ ``for free", and will not specify the concrete algorithm.  Later when we deal with general graphs we will discuss how to find such a direction scheme.
\end{definition}

We still make use of Algorithm~\ref{alg:preprocessing} to partition the graph $G$ into components such that the Cheeger constant of each component is larger than $h = \beta \eps^2$. One issue here is that, after storing and removing those cut edges (denoted by $Q$ in Algorithm~\ref{alg:preprocessing}), the second property of S2-graph may not hold, since the out-degree of some vertices in $G$'s buddy graph $\vec{G}$ will be reduced. Fortunately, the following lemma shows that the number of vertices whose degree will be reduced more than half is small, and we can thus afford to store all their out-going edges. The out-degree of the remaining vertices is within a factor of $4$, thus we can still effectively bound the variance of our estimator. 

\begin{lemma}
\label{alg:improve-preprocessing}
If we run Algorithm~\ref{alg:preprocessing} on an S2-graph $G = (V, E, w)$ with $h = 2^{-\kappa}$, then
\begin{enumerate}
  \item At most $\tilde O(\beta |V|)$ cut edges (i.e., $Q$) will be removed from $G$.
  
  \item There are at most $\tilde{O}(2^{1-{\kappa}}|V|)$ vertices in $G$'s buddy graph $\vec{G}$ which will reduce their out-degrees by more than a half after the removal of $Q$.
\end{enumerate}
\end{lemma}

\begin{proof}
Since $|E| = O(2^{\kappa} \beta |V|)$, $h = 2^{-{\kappa}}$ and $2^{\kappa} \beta = \tilde O(\frac{1}{\eps^2})$, Lemma \ref{lem:boundQ} directly gives the first part. For the second part, note that for each vertex $u$ we have $\woutdeg{u} \ge 2^\kappa \beta$, thus we need to remove at least $2^{\kappa - 1} \beta$ edges to reduce $\woutdeg{u}$ to $2^{\kappa - 1} \beta$. Therefore the number of such vertices is at most $\tilde{O}(\beta\abs{V}/2^{\kappa - 1} \beta) = \tilde{O}(2^{1-{\kappa}}|V|)$.
\end{proof}

For each component $P = (V_P, E_P)$ after running Algorithm~\ref{alg:preprocessing}, let $\vec{P} = (V_P, \vec{E}_P)$ be its buddy directed graph. Slightly abusing the notation, define $\S(\vec{P}) = \{(u, v) \in \vec{E}_P\ |\ \uwoutdegP{u} < 2^{\kappa - 1}\beta\}$, and $\L(\vec{P}) = \{(u, v) \in \vec{E}_P\ |\ \uwoutdegP{u} \ge 2^{\kappa - 1}\beta\}$. We will again omit ``$(\vec{P})$" or ``$(P)$" when there is no confusion. 

The sketch for an S2-graph $G$ is constructed using Algorithm~\ref{alg:improved-simple}. 

\begin{algorithm}[t]
\KwIn{An S2-graph $G$ and a parameter $\eps$}
\KwOut{A $(1+\eps, 0.01)$-spectral-sketch of $G$}
$\{\mathcal{P}, Q\} \gets$ {\bf Preprocessing}($G, \beta \eps^2$)\;
\ForEach{$P = (V_P, E_P, w) \in \mathcal{P}$}{
Let $\vec{P} = (V_P, \vec{E}_P, w)$ be its buddy directed graph\;

\ForEach{$u \in V_P$}{Add $\windegP{u}$ and $\delta_u(P)$  to $\sk{G}$\;}
Add $\S(\vec{P})$ to $\sk{G}$\;

\ForEach{$u \in V_P$}{
Sample $\beta = c_\beta \eps^{-\frac{8}{5}}$ edges with replacement from $\{(v, u) \in \L\}$, where the probability that $(v, u)$ is sampled is $w_{v,u}/\windegP{u}$. Add sampled edges to $\sk{G}$\;
}

Add $Q$ to $\sk{G}$\;
}
\Return $\sk{G}$
\caption{{\bf Spectral-S2}($G, \eps$)}
\label{alg:improved-simple}
\end{algorithm}

It is easy to see that the size of $\sk{G}$ is bounded $\tilde{O}(2^{1-\kappa}\abs{V_P}) \cdot 2^{\kappa-1}\beta + \tilde{O}(\beta \abs{V})  + \tilde{O}(\beta \abs{V}) = \tilde{O}(\beta \abs{V})$, where the first term in LHS is due to the definition of $\S(\vec{P})$ and the second property of Lemma~\ref{alg:improve-preprocessing}.

Let $Y_u^v$ be the random variable denoting the number of times (directed) edge $(v, u)$ 
is sampled when we process the vertex $u$. Clearly, $\E[Y_u^v] = \frac{\beta w_{v,u}}{\windegP{u}} $ and $\var{Y_u^v} \leq \frac{\beta w_{v, u}}{\windegP{u}}$. For a given $x \in \mathbb{R}^{\abs{V_P}}$, we construct the following estimator for each component $P$ using $\sk{G}$. 

\begin{equation}
  \label{eq:estimator-P}
  I_P = \sum_{u \in V_P} x_u^2  \delta_{u}(P) - 2 \sum_{(u, v) \in \S}  x_u x_v w_{u,v} 
        - 2 \sum_{u \in V_P} \frac{\windegP{u}}{\beta}\sum_{(v, u) \in \L} x_u x_v Y_u^v 
\end{equation}
Similar to the analysis in Section~\ref{sec:basic-simple}, it is easy to show that $I_P$ is an unbiased estimator of $x^T L(P) x$ by noticing

\begin{eqnarray*}
   x^T L(P) x 
   &=& \sum_{(u, v) \in \vec{E}_P}(x_u - x_v)^2 w_{u, v}\\
   & = & \sum_{u \in V_P} x_u^2 \delta_u(P) - 2 \sum_{(u, v) \in \S}  x_u x_v w_{u,v} 
        - 2 \sum_{(u, v) \in \L}  x_u x_v w_{u,v}.
\end{eqnarray*} 
We next bound the variance
\begin{eqnarray*}
  \var{I_P}  &=& \var{2\sum_{u\in V_P}\frac{\windegP{u}}{\beta} \sum_{(v, u) \in \L} x_ux_vY_u^v}\\
             &=& 4\sum_{u\in V_P}\frac{\left(\windegP{u}\right)^2}{\beta^2} \sum_{(v, u) \in \L} x_u^2x_v^2\var{Y_u^v}
              \\
             &\leq& 4\sum_{u\in V_P}\frac{\left(\windegP{u}\right)^2}{\beta^2} \sum_{(v, u) \in \L} x_u^2x_v^2\frac{\beta w_{v, u}}{\windegP{u}}
                \quad \left(\var{Y_u^v} \leq \frac{\beta w_{v,u}}{\windegP{u}}\right)\\
             &=& 4\sum_{u\in V_P}x_u^2\frac{\windegP{u}}{\beta} \sum_{(v, u) \in \L} x_v^2w_{v, u}\\
             &\le& 4\sum_{u\in V_P}x_u^2\frac{\windegP{u}}{\beta} \sum_{(v, u) \in \L} x_v^2\cdot 2\gamma
             \quad \left( w_{v, u} \in [\gamma, 2\gamma)\right)\\\\
             &\leq& \frac{4}{\beta}\sum_{u\in V_P}x_u^2\windegP{u} \sum_{(v, u) \in \mathcal{L}} x_v^2\frac{2\woutdegP{v}}{2^{\kappa-1}\beta}
             \quad \left(\uwoutdegP{v} \geq 2^{\kappa - 1}\beta ~\text{and}~  \woutdegP{v}\geq \gamma \cdot \uwoutdegP{v} \right )\\
             &\leq& \frac{16}{2^{\kappa}\beta^2} \sum_{u\in V_P} \delta_u(P) x_u^2\sum_{v\in V_P} \delta_v(P) x_v^2
                \quad \left(\windegP{u} \leq \delta_u(P) ~\text{and}~ \woutdegP{v} \leq \delta_v(P)\right) \\
             &=& \frac{16}{2^{\kappa}\beta^2} \|\hat x\|_2^4,
\end{eqnarray*}
where $\|\hat x\|_2^4 = \|D^{1/2} x\|_2^4 = \sum_{u\in V_P} \delta_u(P) x_u^2\sum_{v\in V_P} \delta_v(P) x_v^2$.

Similar to before we have $$
  \|\hat{x}\|_2^2 = \hat{x}^T \hat{x} \leq \frac{1}{\lambda_1(\tilde L)}\hat{x}^T \tilde L \hat{x} 
                  \overset{by~(\ref{eq:cheeger})}{\leq} \frac{2}{h_G^2}(x^TLx) \overset{h_G > 2^{-\kappa} \text{ by Algorithm~\ref{alg:preprocessing}}}{<} {2} \cdot {2^{2\kappa}} \cdot (x^TLx),
$$
Recall that in an S2-graph, $2^{\kappa} \leq  1/(\beta \eps^2)$, hence 
$$\frac{\var{I_P}}{\left(x^T L(P) x\right)^2} = O\left(\frac{2^{3\kappa}}{\beta^2}\right) = O\left(1/(\beta^{5}\eps^{6})\right) =  O(\eps^2).$$ Setting constant $c_\beta$ large enough in $\beta = c_\beta \eps^2$, by a Chebyshev's inequality, $I_P$ is a $(1+\eps, 0.01)$-approximation to $x^T L(P) x$.  As before, the standard median trick can boost this probability to $1 - \delta$.

\begin{theorem}
\label{thm:improved-simple}
There is a sketching algorithm which given an S2-graph, outputs a $(1+\eps, \delta)$-spectral-sketch of size $\tilde{O}(\abs{V}/\eps^{\frac{8}{5}})$.
\end{theorem}

\subsection{General Graphs}
To deal with general graphs $G = (V, E, w)$ for which the only requirement is $w_{\max}/w_{\min} \le \text{poly}(|V|)$, we try to ``partition" it to poly$\log{\abs{V}}$ subgraphs, each of which is an S2-graph. We note that the partition we used here is {\em not} a simple vertex-partition or edge-partition, as will be evident shortly.

Our first step is to assign each edge in $E$ a direction so that in the later partition step we can partition $G$ to S2-graphs and simultaneously get their buddy directed graphs. The algorithm is described in Algorithm~\ref{alg:assign-direction}.

\begin{algorithm}[t]
\KwIn{A graph $G = (V, E)$ and a parameter $t$}
\KwOut{$\vec{G}=(V, \vec{E})$, a directed graph by assigning each edge in $G$ a direction}
Arbitrarily assign a direction to each edge in $E$, getting $\vec{E}$\;
\While{$\exists (u, v) \in \vec{E}$ s.t. $\uwoutdeg{u} \geq t$ and $\uwoutdeg{v} < t-1$}{
  Change the direction of $(u, v)$\;
}
\Return{$\vec{G}= (V, \vec{E})$}\;
\caption{{\bf Assign-Direction}($G, t$)}
\label{alg:assign-direction}
\end{algorithm}

\begin{lemma}
  \label{lem:assign-direction}
  Given $G = (V, E)$ and $s>1$ as input,  \textbf{Assign-Direction}$(G, s)$ (Algorithm~\ref{alg:assign-direction}) will finally stop and return $\vec{G} = (V, \vec{E})$ with the property that
  for each $(u, v) \in \vec{E}$, $\uwoutdeg{u} < s$ or $\uwoutdeg{v} \geq s - 1$.
\end{lemma}

\begin{proof}
  The second part is trivial according to Algorithm~\ref{alg:assign-direction}. Now we show that {\bf Assign-Direction}$(G, s)$ will finally stop. Let $S = \{(u, v) \in \vec{E}\ |\ \uwoutdeg{u} \geq s ~\text{and}~ \uwoutdeg{v} < s-1 \}$, and $\Delta(S) = \sum_{(u, v)\in S} (\uwoutdeg{u} - \uwoutdeg{v})$. The algorithm stops if and only if $\Delta(S) = 0$. It is easy to see that $\Delta(S)$ is finite for arbitrary $\vec{G}$, thus the algorithm will stop if we can show that $\Delta(S)$ will decrease by \emph{at least} $2$ each time we execute Line $3$. To this end, we only need to show that each execution of Line $3$ will not add any new edge to $S$.

Consider executing Line $3$ on edge $(u, v)$. For $(u, w) \not \in S$, since $(u, v) \in S$, we have $\uwoutdeg{w} \geq s - 1$. Clearly, after executing Line $3$, $(u, w)$ will not be added to $S$ because $\uwoutdeg{w} \geq s - 1$ still holds.
For $(w, v) \not \in S$, since $(u, v) \in S$, we have $\woutdeg{w} < s$. After executing Line $3$, $\uwoutdeg{w} < s$ still holds, hence $(w, v)$ will not be added into $S$.
\end{proof}

For a set of directed edges $\vec{E}$, let $d_u^{\tt{out}}(\vec{E}) \gets |\{v\ |\ (u, v) \in \vec{E}\}|$.  Our partition step is described in Algorithm~\ref{alg:improved-partition}. We first run the spectral sparsification algorithm ~\cite{BSS14}, and then assign directions to each edge. Next, we partition the edges based on their weights, and then partition the directed graph based on the unweighted out-degree of each vertex. Finally, we recursively perform all the above steps on a subgraph induced by a set of edges which have large weights. Notice that the purpose of introducing directions on edges is to assist the edge partitions.

\begin{algorithm}[t]
\KwIn{A graph $G = (V, E, w)$ and a parameter $\eps$}
\KwOut{A set of graph components $\mathcal{P}$}
\If {$|V|< 3$}{
\Return $\{ G \}$\;
}
Run \cite{BSS14} on $G$ with parameter $\eps$, get a spectral sparsifier $G' = (V, E', w')$ with $|E'| =  \eta\cdot \frac{|V|}{\eps^2}$ where $\eta = \tilde{O}(1)$\;
$\tilde{\eps} = \eps/\sqrt{\eta}$, $s \gets  1/{\tilde{\eps}}^2$\;
$\vec{G}' = (V, \vec{E}', w') \gets$ {\bf Assign-Direction}$(G', 2s)$ \;
Partition $\vec{E}'$ into $\vec{E}_j'$'s such that for each $\vec{E}_j'$, all $e \in \vec{E}_j'$ have $w_e \in [2^j, 2^{j+1})$\; 
$\mathcal{P} \gets \emptyset$\;
\ForEach{$\vec{E}_j'$} {
\tcc{Recall $\beta = c_\beta \eps^{-\frac{8}{5}}$ for a large enough constant $c_\beta$}
  Let $\vec{E}_{-\infty} \gets \left\{(u, v) \in \vec{E}_j'\ |\ d_u^{\tt{out}}(\vec{E}_j') < \beta \right\}$\; 
  Let $\vec{E}_i \gets \left\{(u, v) \in \vec{E}_j'\ |\ d_u^{\tt{out}}(\vec{E}_j') \in [2^i \beta, 2^{i+1}\beta) \right\}$ for all $i \geq 0$ ~s.t.~ $2^i\beta \leq s = 1/{\tilde{\eps}}^2$\;
  Add $G(\vec{E}_{-\infty})$ and $G(\vec{E}_i)$ for all $0 \le i \le \log(1/(\tilde{\eps}^2 \beta))$ into $\mathcal{P}$\;
  Remove $\vec{E}_{-\infty}, \vec{E}_i$ from $\vec{E}'$\;
}
\tcc{Recursively apply on the remaining edges $E'$ (remove  directions on edges)}
\Return{$\mathcal{P} \cup \text{\bf Partition } (G(E'))$}\;
\caption{{\bf Partition}($G$)}
\label{alg:improved-partition}
\end{algorithm}

For the analysis, we first show that after each recursion in Algorithm~\ref{alg:improved-partition}, the number of vertices of the graph induced by the remaining edges will decrease by at least a constant fraction. In this way we can bound the number of recursion steps by $O(\log\abs{V})$.

\begin{lemma}
  \label{lem:recursion}
Given a graph $G = (V, E)$ with $|E| \leq s|V|\ (s > 1)$, let $\vec{G} \gets$ \textbf{Assign-Direction}$(G, 2s)$. If we remove all $(u, v)$ with $d_u^{\tt{out}}(\vec{E}) < 2s$ from $\vec{E}$ and get a subset $\vec{E}_r\subset \vec{E}$, then we have $|V_r(\vec{E}_r)| \leq |V| / (2 - 1/s)$ 
\end{lemma}

Before proving this lemma, note that in Algorithm~\ref{alg:improved-partition}, $|E| \leq s|V|$ is guaranteed by Line $3$, and ``remove all $(u, v)$ with $d_u^{\tt{out}}(\vec{E}) < 2s$" is done by Line $12$.

\begin{proof}
By Lemma \ref{lem:assign-direction}, 
  for each $(u, v) \in \vec{E}$, we have $d_u^{\tt{out}}(\vec{E}) < 2s$ or $d_v^{\tt{out}}(\vec{E}) \geq 2s - 1$.  If we remove all $(u, v)$ with $d_u^{\tt{out}}(\vec{E}) < 2s$, then for each $(u, v) \in \vec{E}_r$, we have $d_u^{\tt{out}}(\vec{E}_r) \geq 2s$ and $d_v^{\tt{out}}(\vec{E}_r) \ge 2s-1$. Consequently, $|V_r(\vec{E}_r)| (2s - 1) \leq |E| \leq s|V|$. Therefore we have
  $|V_r(\vec{E}_r)| \leq |V| / (2 - 1/s)$.
\end{proof}

The following lemma summarizes the properties of $\mathcal{P}$ returned by Algorithm~\ref{alg:improved-partition}.

\begin{lemma}
  \label{lem:improved-partition}
  Given $G = (V, E, w)$ with $w_{\max}/w_{\min} = \text{poly}(|V|)$, let $\mathcal{P} \gets$~\textbf{Partition}$(G)$ be a set of graphs after the partition, then (1) $|\mathcal{P}| = \text{poly}(\log |V|)$; and (2) for each $\vec{P} = (V_P, \vec{E}_P, w_P) \in \mathcal{P}$, if $|V_P| > 2$ then for any $e \in \vec{E}_P$,  $w_e \in [\gamma, 2\gamma)$ for some $\gamma > 0$, and one of the following properties holds:
  \begin{enumerate}
  \item[] {\em Property $1$:} For each $u \in V_P$, $\uwoutdegP{u} < \beta$.
  \item[] {\em Property $2$:} There exists $i\ (0 \le i \le \log(\eta/(\beta \eps^2)))$,
    for each $u \in V_P$, $\uwoutdegP{u} \in [2^i\beta, 2^{i+1}\beta)$.
  \end{enumerate}
\end{lemma}

\begin{proof}
We only need to bound the size of $\mathcal{P}$. The rest directly follows from the algorithm.

First, Line $6$ and Line $10$ will partition $\vec{E}'$ to $O(\log^2\abs{V})$ (assuming $\abs{V} > 1/\eps$) sets. 
Second, we bound the number of recursion steps. Note that if we directly remove $E_s = \{(u, v) \in \vec{E}'\ |\ d_u^{\tt{out}}(\vec{E}') < 2s\}$ from $\vec{E}'$, then by Lemma~\ref{lem:recursion} we know that there are at most $O(\log\abs{V})$ recursion steps. The subtlety is that we first partition $\vec{E}'$ into $\vec{E}'_j$'s and then remove all $(u, v) \in \vec{E}'_j$ with $d_u^{\tt{out}}(\vec{E}'_j) < 2s$. However, since $d_u^{\tt{out}}(\vec{E}'_j) \le d_u^{\tt{out}}(\vec{E}')$, every edge in $E_s$ will still be removed by at Line $12$. Therefore $\abs{\mathcal{P}} = O(\log^3\abs{V})$.
\end{proof}

We summarize our main result.

\begin{theorem}
  \label{thm:improved-general}
  Given an undirected positively weighted $G = (V, E, w)$ with $w_{\max}/w_{\min} \le \text{poly}(|V|)$, there is a sketching algorithm that outputs a $(1+\eps, 0.01)$-spectral-sketch of size $\tilde{O}(\abs{V}/\eps^{\frac{8}{5}})$.
\end{theorem}

\begin{proof}
Our final algorithm is described in Algorithm~\ref{alg:sketch-improved-general}.
For each component $H \in \mathcal{P}$, we store the whole $H$ if Property $1$ in Lemma~\ref{lem:improved-partition} holds (let $\mathcal{P}_1$ be this set), and we apply Theorem~\ref{thm:improved-simple} on $H$ (set $\delta = 1/\text{poly}\log \abs{V}$) if Property $2$  in Lemma~\ref{lem:improved-partition} holds (let $\mathcal{P}_2$ be this set), thus the space usage is bounded by $\tilde{O}(\abs{V}/\tilde{\eps}^{\frac{8}{5}}) = \tilde{O}(\abs{V}/\eps^{\frac{8}{5}})$.  Since by $\abs{\mathcal{P}_1 \cup \mathcal{P}_2} = \abs{\mathcal{P}} = O(\log^3\abs{V})$ by Lemma~\ref{lem:improved-partition}, we can bound the total space by $\tilde{O}(\abs{V}/\eps^{\frac{8}{5}})$.

Our final estimator, given $x \in \mathbb{R}^{\abs{V}}$, is 
$$I_G = \sum_{H \in \mathcal{P}_2} I_H + \sum_{H \in \mathcal{P}_1} x^T L(H) x,$$ where $I_H$ is defined in Equation~(\ref{eq:estimator-P}). Similar to the proof of Lemma~\ref{lem:basic-general}, we can bound $\var{I_G}$ by $O(\eps^2 \cdot (x^T L(G) x)^2)$. By a Chebyshev's inequality, $I_G$ is a $(1+\eps, 0.01)$-approximation of $x^T L(G) x$.
\end{proof}

\begin{algorithm}[t]
\KwIn{$G = (V, E, w)$; a quality control parameter $\eps$}
\KwOut{A $(1+\eps, 0.01)$-spectral-sketch $\sk{G}$ of $G$}
Let $\mathcal{H} \gets$ {\bf Partition}(G)\;
$\sk{G} \gets \emptyset$\;
\ForEach{$H \in \mathcal{H}$} {
\If{$H$ satisfies Property $1$ in Lemma~\ref{lem:improved-partition}}{
     Add the whole $H$ to $\sk{G}$\;
}


\ElseIf{$H$ satisfies Property $2$ in Lemma~\ref{lem:improved-partition}}{Add \textbf{Spectral-S2}($H, \eps$) into $\sk{G}$ \;}
}

\Return{$\sk{G}$}\;
\caption{{\bf Spectral-Improved}($G, \eps$)}
\label{alg:sketch-improved-general}
\end{algorithm}

{\bf Acknowledgments:} David Woodruff would like to thank Alex Andoni, Jon Kelner, and Robi Krauthgamer for helpful discussions, as well as acknowledge support from the XDATA program of the Defense Advanced Research Project Agency (DARPA), administered through Air Force Research Laboratory contract FA8750-12-C-0323. Qin Zhang would like to thank the organizers of Bertinoro Workshop on Sublinear Algorithms 2014 where he heard this problem.
\bibliographystyle{abbrv}
\bibliography{paper}

\newpage

\appendix
\section{Appendix}
\subsection{Lower Bound for General Matrices with the For All Guarantee}
\begin{theorem}
\label{thm:all-exact}
Any sketch $s(A)$ which satisfies the for all guarantee, even when all of the entries are promised to be in the range $\{-1, -1+1/n^C, -1+2/n^C, \ldots, 1-1/n^C, 1\}$ for a sufficiently large constant $C > 0$, 
must use $\Omega(n^2)$ bits of space.
\end{theorem}
\begin{proof}
We follow the outline given in the introduction. 

Consider a net of $n \times n$ projector matrices
$P$ onto $n/2$-dimensional subspaces $U$ of $\mathbb{R}^n$. It is known 
(see Corollary 5.1 of \cite{kt13}, which uses a result of \cite{aek04}) that there
exists a family $\mathcal{F}$ of $r = 2^{\Omega(n^2)}$ distinct matrices 
$P_{1}, \ldots, P_{r}$ so that for all $i \neq j$, 
$\|P_i - P_j\|_2 \geq \frac{1}{2}$. By rounding each
of the entries of each matrix $P_i$ to the nearest additive 
multiple of $1/n^C$,
obtaining a symmetric matrix $Q_i$ with entries 
in $\{-1, -1+1/n^C, \ldots, 1\}$, 
we obtain a family $\mathcal{F}'$ of $r=2^{\Omega(n^2)}$ distinct matrices
$Q_i$ such that $\|Q_i - Q_j\|_2 \geq \frac{1}{4}$ for all $i \neq j$. This
implies there is a unit vector $x^*$ for which
$\|Q_ix^*-Q_jx^*\|_2 \geq \frac{1}{4}$, or equivalently, 
\begin{eqnarray}\label{eqn:optimize}
\|Q_ix^*\|_2^2 + \|Q_jx^*\|_2^2 - 2 \langle (x^*)^TQ_i^T, Q_j x^* \rangle \geq \frac{1}{16}.
\end{eqnarray}
Let $J$ be the subspace of $\mathbb{R}^n$ which is the intersection of the spaces
spanned by $Q_i$ and $Q_j$, let $K_i$ be the subspace of $Q_i$ orthogonal to $J$,
and let $K_j$ be the subspace of $Q_j$ orthogonal to $J$. We identify $J$, $K_i$, and $K_j$,
with their corresponding projection matrices. To maximize the left hand side of 
(\ref{eqn:optimize}), we can assume the unit vector $x^*$ is in the span of the union of $J, K_i,
$ and $K_j$. We can therefore write $x^* = Jx^* + K_ix^* + K_j x^*$, and note that the three summand
vectors are orthogonal to each other. Expanding (\ref{eqn:optimize}),
the left hand side is equal to 
$$2\|Jx^*\|_2^2 + \|K_ix^*\|_2^2 + \|K_j x^*\|_2^2
- 2 \|Jx^*\|_2^2 =  \|K_ix^*\|_2^2 + \|K_j x^*\|_2^2.$$
Hence, by (\ref{eqn:optimize}), it must be that either $\|K_ix^*\|_2^2 \geq \frac{1}{32}$ or
$\|K_j x^*\|_2^2 \geq \frac{1}{32}$. This implies the vector $z = K_ix^*$ satisfies
$\|Q_iz\|_2^2 \geq \frac{1}{32}$, but $\|Q_jz\|_2^2 = 0$. 

Therefore,
if there were a sketch
$s(A)$ which had the ``for all'' guarantee for any matrix $A \in \mathcal{F}$, 
one could query $s(A)$ on the vector $z$ given above for each pair $Q_i, Q_j \in \mathcal{F}$, 
thereby recovering the matrix $A \in \mathcal{F}$. Hence, $s(A)$ is an encoding of an arbitrary
element $A \in \mathcal{F}$ which implies that the size of $s(A)$ is 
$\Omega(\log |\mathcal{F}|) = \Omega(n^2)$ bits, completing the proof. 
\end{proof}

\subsection{Lower Bound for General Matrices with the For Each Guarantee}
\begin{theorem}
\label{thm:all-approx}
Any sketch $s(A)$ which satisfies the for each guarantee with $(1+\eps)$-approximation with constant probability, 
even when all of the entries of $A$ are promised to be in the set $\{0, 1\}$, 
must use $\Omega(n/\eps^2)$ bits of space. 
\end{theorem}

The proof is very similar (and inspired from) a result in \cite{GWWZ14} for approximating the number of non-zero entries of $Ax$. We include it here for completeness. Before giving the proof, we need some tools. Let $\Delta(a,b)$ be the Hamming distance between two bitstrings $a$ and $b$.

\begin{lemma}[modified from \cite{JKS08}]
\label{lem:gap-ham}
Let $x$ be a random bitstring of length $\gamma = 1/\eps^2$, and let $i$ be a random index in $[\gamma]$. Choose $\gamma$ public random bitstrings $r^1, \ldots, r^{\gamma}$, each of length $\gamma$. Create $\gamma$-length bitstrings $a$, $b$ as follows:
\begin{itemize}
\item
For each $j \in [\gamma]$, $a_j = \text{majority}\{r^j_k\ |\ \text{indices } k \text{ for which } x_k = 1\}$.
\item
For each $j \in [\gamma]$, $b_j = r^j_i$.
\end{itemize}
There is a procedure which, with probability $1/2+\delta$ for a constant $\delta > 0$, can determine the value of $x_i$ from any $c \sqrt{\Delta(a,b)}$-additive approximation to $\Delta(a,b)$, provided $c > 0$ is a sufficiently small constant.
\end{lemma}

We introduce the \IND\ problem. In \IND, we have two randomized 
parties Alice and Bob. Alice has $x \in \{0, 1\}^n$, and Bob has an index $i \in [n]$. The communication is one-way from Alice to Bob, and the goal is for Bob to compute $x_i$.

\begin{lemma}[see, e.g., \cite{KN97}]
\label{lem:index}
To solve the \IND\ problem with success probability $1/2+\delta$ for any constant $\delta > 0$, Alice needs to send Bob $\Omega(n)$ bits  even with shared randomness.
\end{lemma}

\begin{proof} (for Theorem~\ref{thm:all-approx})
Let $\gamma = 1/\eps^2$. The proof is by a reduction from the \IND\ problem, where Alice has a random bitstring $z$ of length $(n - \gamma) \cdot \gamma$, and Bob has an index $\ell \in [(n - \gamma) \cdot \gamma]$. 

Partition $z$ into $n - \gamma$ contiguous substrings $z^1, z^2, \ldots, z^{n - \gamma}$. Alice constructs a matrix $A$ as follows: she uses shared randomness to sample $\gamma$ random bitstrings $r^1, \ldots, r^\gamma$, each of length $\gamma$. For the leftmost $\gamma \times \gamma$ submatrix of $A$, in the $i$-th column for each $i \in [\gamma]$, Alice uses $r^1, \ldots, r^\gamma$ and the value $i$ to create the $\gamma$-length bitstring $b$ according to Lemma~\ref{lem:gap-ham} and assigns it to this column. Next, in the remaining $n - \gamma$ columns of $A$, in the $j$-th column for each $j \in \{\gamma+1, \ldots, n\}$, Alice uses $z^{j - \gamma}$ and $r^1, \ldots, r^\gamma$ to create the  $\gamma$-length bitstring $a$ according to Lemma~\ref{lem:gap-ham} and assigns it to this column.

Alice then sends Bob the sketch $s(A)$, together with $\{\|A^i\|_2^2 \ (i \in [n])\}$ and $\{nnz(A^i) \ (i \in [n])\}$ where $nnz(x)$ is the number of non-zero coordinates of $x$. Note that both $\{\|A^i\|_2^2 \ (i \in [n])\}$ and $\{nnz(A^i) \ (i \in [n])\}$ can be conveyed using  $O(n \log (1/\eps)) = o(n/\eps^2)$ bits of communication, which is negligible.

Bob creates a vector $x$ by putting a $1$ in the $i$-th and $j$-th coordinates, where $i, j\ (i \in [\gamma], j \in \{\gamma + 1, \ldots,  n\})$ satisfies $\ell = i + (j - \gamma - 1) \cdot \gamma$. Then $Ax$ is simply the sum of the $i$-th and $j$-th columns of $A$, denoted by $A^i + A^j$. Note that $A^i, A^j$ correspond to a pair of $(a, b)$ created from $r^1, \ldots, r^\gamma$ and $z^{j - \gamma}$ (and $(z^{j - \gamma})_i = z_\ell$). Now, from a $(1 + c_\eps\eps)$-approximation to $\|A^i + A^j\|_2^2$ for a sufficiently small constant $c_\eps$ (which is an absolute constant, independent of $\eps$), and exact values of $\|A^i\|_2^2$ and  $\|A^j\|_2^2$, Bob can approximate $2 \langle A^i, A^j \rangle= \|A^i + A^j\|_2^2 - \|A^i\|_2^2 - \|A^j\|_2^2$ up to an additive error $c_\eps \eps \cdot \|A^i + A^j\|_2^2 \le c'_\eps /\eps$ for a sufficiently small constant $c'_\eps$. Then, using a $c'_\eps /\eps$-additive approximation of $2 \langle A^i, A^j \rangle$, and exact values of $nnz(A^i)$ and $nnz(A^j)$, Bob can approximate $\Delta(A^i, A^j) = nnz(A^i) + nnz(A^j) - 2 \langle A^i, A^j \rangle$ up to a $c'_\eps/\eps = c''_\eps \sqrt{\Delta(A^i, A^j)}$ additive approximation for a sufficiently small constant $c''_\eps$, and consequently compute $z_\ell$ correctly with probability $1/2+\delta$ for a constant $\delta > 0$ (by Lemma~\ref{lem:gap-ham}).  

Therefore, any algorithm that produces a $(1+c_\eps \eps)$-approximation of $\|A^i + A^j\|_2^2$ with probability $1 - \delta'$ for some sufficiently small constant $\delta' < \delta$ can be used to solve the Indexing problem of size $(n - \gamma) \gamma = \Omega(n/\eps^2)$ with probability $1 - \delta' - (1/2 - \delta) > 1/2 + \delta''$ for a constant $\delta'' > 0$. The theorem follows by the reduction and Lemma~\ref{lem:index}. 
\end{proof}

\end{document}